\documentclass[11pt]{article}
\usepackage{fullpage}
\usepackage{url}
\usepackage[colorlinks=true,linkcolor=Emerald,citecolor=RoyalBlue]{hyperref}
\usepackage{amsmath,nccmath, amsfonts,amsthm,multirow,mdframed,amssymb}
\usepackage{times}
\usepackage{paralist}
\usepackage{graphicx}
\usepackage{floatpag}
\usepackage{float}
\usepackage{bbold}
\usepackage{enumitem}
\usepackage{subcaption} 
\usepackage{soul}
\usepackage{comment}
\usepackage{calc}
\usepackage{upgreek}
\usepackage{mathtools}
\usepackage{color}
\usepackage[dvipsnames]{xcolor}
\usepackage{xspace}
\usepackage{tcolorbox}
\usepackage{cleveref}
\usepackage{stackengine}
\stackMath
\allowdisplaybreaks

\def \N {\mathbb{N}}

\def \cD {\mathcal{D}}
\def \cE {\mathcal{E}}

\def \cP {\mathcal{P}}
\def \cR {\mathcal{R}}
\def \cS {\mathcal{S}}

\def \tO {\widetilde{O}}

\def \poly {\mathrm{poly}}

\def \eps {{\varepsilon}}

\def\ggg{\gtrsim}
\def\lll{\lesssim}

\def \deg {\text{deg}}

\def \out {\textsc{out}}
\def \In {\textsc{in}}
\def \pl {\sf{polylog}}
\def \pll {\text{polyloglog}}

\renewcommand{\leq}{\leqslant}

\renewcommand{\geq}{\geqslant}
\renewcommand{\ge}{\geqslant}

\usepackage{tikz}
\newcommand\rcirc{\tikz[baseline=(char.base)]{
		\node[shape=circle,draw,inner sep=0.8pt] (char) {\small r};}}

\newtheorem{theorem}{Theorem}[section]

\newtheorem{proposition}[theorem]{Proposition}

\newtheorem{corollary}[theorem]{Corollary}

\newtheorem{lemma}[theorem]{Lemma}
\newtheorem*{lemma*}{Lemma}
\newtheorem*{theorem*}{Theorem}
\newtheorem{claim}[theorem]{Claim}
\newtheorem*{claim*}{Claim}

\newtheorem{remark}[theorem]{Remark}
\newtheorem*{remark*}{Remark}
\newtheorem{definition}[theorem]{Definition}
\theoremstyle{definition}
\newtheorem{algo}{Algorithm}
\newtheorem{protocol}{Protocol}
\newtheorem{agr-test}{Agreement-Test}
\newtheorem{list-agr-test}{List-Agreement-Test}

\makeatletter
\let\c@fconjecture\c@conjecture
\makeatother

\makeatletter
\let\c@fconj\c@conj
\makeatother

\title{Constant Degree Networks for \\ Almost-Everywhere Reliable Transmission}

\author{Mitali Bafna \thanks{Department of Mathematics, Massachusetts Institute of Technology.}
	\and
	Dor Minzer\thanks{Department of Mathematics, Massachusetts Institute of Technology. Supported by NSF CCF award 2227876 and NSF CAREER award 2239160.}
}

\date{\vspace{-5ex}}
\begin{document}
	\maketitle
	
	\begin{abstract}
		In the almost-everywhere reliable message transmission problem, introduced in \cite{Dwork}, the goal is to design a sparse communication network $G$ that supports efficient, fault-tolerant protocols for interactions between all node pairs. By fault-tolerant, we mean that that even if an adversary corrupts a small fraction of vertices in $G$, then all but a small fraction of vertices can still communicate perfectly via the constructed protocols. Being successful to do so allows one to simulate, on a sparse graph, any fault-tolerant distributed computing task and secure multi-party computation protocols built for a complete network, with only minimal overhead in efficiency. Previous works on this problem \cite{Dwork, Upfal, Chandran, JRV, BMV24} achieved either constant-degree networks tolerating $o(1)$ faults, constant-degree networks tolerating a constant fraction of faults via inefficient protocols (exponential work complexity), or poly-logarithmic degree networks tolerating a constant fraction of faults. 
		
		We show a construction of constant-degree networks with efficient protocols (i.e., with polylogarithmic work complexity) that can tolerate a constant fraction of adversarial faults, thus solving the main open problem of~\cite{Dwork}. Our main contribution is a composition technique for communication networks, based on graph products. Our technique combines two networks tolerant to adversarial \emph{edge-faults} to construct a network with a smaller degree while maintaining efficiency and fault-tolerance. We apply this composition result multiple times, using the polylogarithmic-degree edge-fault tolerant networks constructed in~\cite{BMV24} (that are based on high-dimensional expanders) with itself, and then with the constant-degree networks (albeit with inefficient protocols) of~\cite{Upfal}.
	\end{abstract}

	\section{Introduction}
Many real-world applications involve computations on inputs that might be distributed across many machines in a large network. This need has led to the development of protocols for important distributed tasks like Byzantine agreement~\cite{LamportSP82}, collective coin flipping, poker and more generally, for secure multiparty computation, which also ensures privacy in addition to correct computation. In fact, this culminated in the completeness theorems of \cite{Ben-OrGW88, ChaumCD88} showing that any joint function can be computed even with a constant fraction of Byzantine parties--those whose behavior might deviate arbitrarily from the protocol--while ensuring correctness and privacy.

Most of these protocols assume that each machine can directly communicate with every other machine in the network. However, such assumptions are impractical for modern large-scale networks, which are often sparsely connected. To address this, the seminal work of Dwork, Peleg, Pippenger, and Upfal \cite{Dwork} studied the question of designing sparse networks that are resilient to Byzantine node failures. Their goal was to design a sparse network $G$ of degree $d$, on $n$ nodes, where honest nodes can still communicate and execute protocols even if $t$ nodes are adversarially corrupted. Since $t$ might be much larger than $d$, some honest nodes may become isolated if all of their neighbors are corrupted. Therefore, Dwork et al.~allow $x$ (possibly larger than $t$) nodes to become ``doomed'', requiring only the remaining $n - x$ nodes to successfully participate in all protocols between them. 

They proposed the almost-everywhere reliable message transmission problem as a way to simulate any fault-tolerant protocol on the complete network. The problem asks one to construct an $n$-vertex sparse communication network $G = (V,E)$, along with a set of efficient communication protocols $\cR=\{R(u,v)\}_{u,v\in G}$ between pairs of vertices of $G$ that are fault-tolerant. Formally, we are interested in the following parameters of $G$ and $\cR$:
\begin{enumerate}
\item Sparsity: The degree of $G$. Ideally, we would like it to be an absolute constant, independent of $n$.
\item Round complexity: The number of rounds of communication, where in each round every vertex can send and receive messages from its neighbors in $G$. We remark that trivially, the number of rounds is at least the diameter of the graph, and in particular it is at least $\Omega(\log n)$ if $G$ has a constant degree.
\item Work complexity: The work complexity of a protocol is the total computation performed by all vertices to implement it. The work complexity of $\cR$ is the maximum over the work complexity of the protocols $R(u,v)$. Ideally, we would like this to be $\pl n$.
\item $(\eps,\nu)$-Fault-Tolerance: If an adversary corrupts any $\eps$-fraction of the nodes of $G$ then all but $\nu n$ nodes, referred to as doomed nodes, can communicate perfectly with each other using the protocols in $\cR$. Ideally, we would like to be able to take $\eps$ constant bounded away from $0$, and $\nu$ as a vanishing function of $\eps$, say $\nu = \eps^{\Omega(1)}$.
\end{enumerate}
The work~\cite{Dwork} gave a construction of constant degree networks with protocols that have work complexity $\pl(n)$ and tolerance parameters $\eps = 1/\log n$, $\nu = \Omega(1)$. They also showed how to simulate any protocol built for the complete graph (that is tolerant to $\nu n$-fraction of vertex corruptions), using communication on the edges of the sparse graph $G$. The resulting protocol on $G$ is tolerant to $\eps n$ adversarial corruptions, and ensures that all but $\nu n$-fraction of the parties compute the desired output. In~\cite{GarayO08, ChandranFGOPZ22} this was extended to secure multiparty computation (MPC), in that, given $G$ as above, they showed how to construct a related sparse network that supports almost-everywhere secure MPC. Following~\cite{Dwork} the works~\cite{Upfal, Chandran, chandran2012edge, JRV, BMV24} improved parameters for the a.e.~transmission problem. While these works obtained optimality in a few parameters, they did not achieve all of the ideal parameters simultaneously. Our main result is a sparse network achieving it: our network has constant degree, and our communication protocols have polylogarithmic round and work complexity and constant fault tolerance. This resolves the main open problem of~\cite{Dwork}. 

At the heart of our proof is a composition technique for communication protocols, reminiscent of the composition technique from the theory of probabilistically checkable proofs (PCPs)~\cite{FGLSS,AroraSafra,ALMSS} and of expander graphs~\cite{ReingoldVW00}. Our composition result (\Cref{lem:composition}) shows how to compose two networks that are tolerant to \emph{edge-faults} in the network and obtain a new network with smaller degree, while maintaining the tolerance and work-complexity. Here and throughout, the edge-fault model is the model in which the adversary is allowed to corrupt an arbitrary $\eps$-fraction of the edges of the graph, 
and the $(\eps,\nu)$-fault-tolerant requirement is  
that at most $\nu$ fraction of the vertices are doomed.\footnote{Note that in the edge-fault model, an adversary can corrupt any small fraction of edges, thus it is stronger than the adversary in the adversary in the vertex corruption model. Indeed, to mimic a vertex corruption, the edge-fault adversary can choose to control all edges adjacent to a vertex.}

Interestingly, we do not know whether an analog of our composition statement holds if the networks $G$ and $H$ are only promised to be tolerant against vertex corruptions. While the two models are qualitatively the same on constant degree graphs (up to constant factors), results in the edge-fault model are substantially more powerful for graphs of super constant degree. 
For instance, the prior state-of-the-art results in literature~\cite{Chandran, JRV} constructed (poly)logarithmic-degree graphs that are tolerant against constant-fraction of vertex-faults, but are not tolerant against a constant-fraction of edge-faults. The work~\cite{BMV24} is the first work that constructs polylogarithmic degree networks with routing protocols that are efficient and tolerant against constant fraction of edge corruptions.

We instantiate our composition theorem by composing the network  of~\cite{BMV24} with itself a few times, until the degree becomes sufficiently small (but still super constant). Then, we compose the resulting network with the constant degree networks of~\cite{Upfal} to reduce the degree to a constant, while maintaining the constant tolerance and polylogarithmic work complexity.

\subsection{Our Results}\label{sec:results}
We now state our main result formally.
\begin{theorem}\label{cor:main}
There exists $D\in \N$ such that for all $\eps>0$ and large enough $n$, there exists a $D$-regular graph $G$ with $\Theta(n)$ vertices and a set of protocols $\cR=\{R(u,v)\}_{u,v\in G}$ between all pairs of vertices in $G$, with work complexity $\pl n$ and round complexity $\tO(\log n)$, such that if at most $\eps$-fraction of edges are corrupted, then at most $\poly(\eps)$-fraction of vertices in $G$ are doomed. Furthermore, there is a polynomial time deterministic algorithm to compute $G$, and a randomized algorithm to construct the protocols $\cR$, which is guaranteed to satisfy the above tolerance guarantees with probability $1-\exp(-n\pl n)$. 
\end{theorem}

To prove~\Cref{cor:main}, we first establish a version of the above result in the ``permutation model'' of routing (which is the model most relevant for PCP constructions). In this model, given a permutation $\pi$ on $V(G)$, the goal is to construct a set of protocols $\cR=\{R(u,\pi(u))\}_{u\in G}$, such that if at most $\eps$-fraction of behave adversarially, then at most $f(\eps)$-fraction of the  protocols in $\cR$ fail, denoted by $(\eps,f(\eps))$-tolerance; see~\Cref{def:perm-model} for a formal definition.
\begin{lemma}\label{thm:main}
There exists $D>0$ such that for all $\eps>0$, and all $n$, there exists a $D$-regular graph $G = (V,E)$ with $\Theta(n)$ vertices such that for each permutation $\pi\colon V\to V$, the graph $G$ admits protocols $\mathcal{R} = \{R(u,\pi(u))\}_{u\in G}$ with work complexity $\pl n$ that are $(\eps,\poly(\eps))$-tolerant. Both the graph and the protocols can be constructed deterministically in polynomial time.
\end{lemma}
We prove~\Cref{thm:main} in~\Cref{sec:perm-final}, 
and in~\Cref{sec:cor-final} we show how to deduce~\Cref{cor:main} from it.\footnote{We remark that the two models discussed above are morally equivalent:~\Cref{cor:main} also implies~\Cref{thm:main}, albeit with a randomized algorithm for constructing  $\cR$.}

\subsection{Implications}\label{sec:implications}
We now discuss the implications of~\Cref{cor:main}. Firstly, it implies that one can simulate any protocol built for the complete network on the sparse networks from~\Cref{thm:main}, albeit with polylogarithmic overhead. This is because any message transfer on the edge $(u,v)$ in the complete graph, can be simulated using the protocol $R(u,v)$ for the sparse network. More precisely:
\begin{corollary} \label{cor:MPP}
Let $G$ be a $D$-regular graph on $n$ vertices from~\Cref{thm:main} with the set of protocols $\cR = \{R(u,v)\}_{u,v\in V(G)}$, then for all $\eps > 0$ the following holds. Suppose $P$ is a protocol for a distributed task $T$ on the complete network on $n$ nodes that can tolerate up to $\eps^c n$ corrupted nodes, where $c>0$ is a universal constant. Then, there exists a protocol $\Pi$ such that:
\begin{enumerate}
\item The number of rounds of communication in $\Pi$ is at most $\text{round}(P)\pl n$, where $\text{round}(P)$ denotes the round complexity of $P$.
\item The work complexity of $\Pi$ is at most $\text{work}(P)\pl n$, where $\text{work}(P)$ denotes the work complexity of $P$, that is, the total computation performed by all nodes to implement $P$.
\item  If an adversary corrupts any $\eps$-fraction of nodes, then there is a set of nodes $S \subset V$ where $|S| \ge (1-\eps^c)n$ that output the desired value as required by the task $T$.
\end{enumerate}
\end{corollary}
One can instantiate~\Cref{cor:MPP} with protocols for Byzantine agreement
to get sparse networks that support protocols for almost-everywhere agreement. 
Using the result of~\cite{GarayO08}, we immediately get sparse networks that support almost-everywhere secure MPC. We refer the reader to~\cite[Theorem 4.3]{GarayO08} and further developments by~\cite{ChandranFGOPZ22} for a formal treatment of the topic.

\subsection{Techniques}\label{sec:techniques}
The proof of~\Cref{thm:main} consists of several building blocks. First, we need a construction of  network with the properties as outlined in the theorem, except that we allow the degree of $G$ to be poly-logarithmic in $n$. This component is achieved by the construction of~\cite{BMV24}. Second, we need a network with constant degree which is tolerant to a constant fraction of edge corruptions, however it is allowed to be inefficient. This component is achieved by the construction of~\cite{Upfal}. Finally, we need a composition technique which allows one to combine two networks with fault-tolerant routing protocols, one large and another small one, to get a network that inherits the degree of the smaller network and has fault-tolerant routing protocols. This is shown in~\Cref{lem:composition} and is the main technical contribution of the current work. 

In the rest of the technical overview, we elaborate on the proof of~\Cref{lem:composition}, and at the end we discuss how to use this result to achieve~\Cref{thm:main}.
\vspace{-1ex}
\subsubsection{The Balanced Replacement Product}\label{sec:rep-product} 
Our composition uses a graph product known as the balanced replacement product~\cite{ReingoldVW00}. Given an $n$-vertex, $d$-regular graph $G$ and a $d$-vertex, $k$-regular graph $H$, consider the graph $Z = G\rcirc H$ defined as follows. First, for each vertex of $G$, fix some ordering on the $d$ edges incident to it. 
\begin{itemize}
  \item Replace a vertex $u$ of $G$ with a copy of the graph $H$, henceforth denoted by $C_u$ and referred to as the cloud of $u$. For $u \in V(G)$, $c \in V(H)$, let $(u,c) \in V(Z)$ denote the $c^{th}$ vertex in the cloud of $u$. Also, let $(u,c)_1$ denote $u\in V(G)$ and $(u,c)_2$ denote $c\in V(H)$.
  \item Associate each edge $e\in E(G)$ incident on $u$ to a unique vertex in the cloud of $u$, according to the ordering of the edges incident on $u$. 
  \item For an edge $e \in E(G)$ with endpoints $u,v$, we add $\deg(H)$ parallel edges between the vertices $(u,e_1)$ and $(v,e_2)$, where these are the vertices
 in $C_u$ and $C_v$ that are associated to the edge $e$. %
\end{itemize}
Note that $Z$ has $|V(G)||V(H)|$ vertices (using the fact that $\deg(G) = |V(H)|$) and each vertex $(u,a)$ has $\deg(H)$ edges incident on it inside the cloud $C_u$ and $\deg(H)$ edges incident on it outside the cloud. This implies that $Z$ has degree $2\deg(H)$ and also that the total number of edges inside clouds is equal to the total number of edges across clouds. The utility of this fact will be that if at most $\eps$-fraction of edges in $G\rcirc H$ have been corrupted, then at most $2\eps$ of the edges inside the clouds and at most $2\eps$
of the edges across the clouds have been corrupted.

\subsubsection{Routing Protocols on the Composed Graph} 
The next step in our proof is to design fault-tolerant routing protocols on $Z=G\rcirc H$, given a set of routing protocols for $G$ and $H$ that are tolerant against a constant fraction of edge corruptions. The work complexity of the new protocols is roughly the product of the work complexities of the protocols on $G$ and $H$. Thus, as the degree of $Z$ is $2\deg(H)$, we have produced a network with similar performance to $G$ but much smaller degree. 

Our composition result is proved for the permutation model (towards proving~\Cref{thm:main}). Henceforth assume that for every permutation $\pi'$ on $V(G)$, there is a set of fault-tolerant protocols $\{R(u,\pi'(u))\}_{u\in G}$ and similarly fix a set of fault-tolerant protocols $\{P(a,b)\}_{a,b\in V(H)}$ between all pairs of vertices of $H$. Now, given any permutation $\pi\colon V(Z)\to V(Z)$, we will design routing protocols $\cR = \{R((u,a),\pi(u,a))\}_{(u,a)\in Z}$ whose fault-tolerance parameters depend on the tolerance parameters of $G$ and $H$.

At a high level, each protocol $R((u,a),\pi(u,a))$ simulates some protocol $R$ of $G$ for a message transfer from $u$ to $\pi(u,a)_1$; note that $u$ and $\pi(u,a)_1$ are the clouds that $(u,a)$ and $\pi(u,a)$ belong to respectively. To decide which protocol on $G$ to use for each transmission from $(u,a)\rightarrow \pi(u,a)$, we first ``decompose'' the permutation $\pi$ into $d=\deg(G)$ permutations $\pi_1,\ldots, \pi_d$ on $V(G)$ (see~\Cref{sec:composition}). We then invoke the premise on the graph $G$, asserting that there are protocols $\mathcal{P}_i= \{R(u,\pi_i(u))\}$ corresponding to each matching $\pi_i$. Now the protocol $R((u,a),\pi(u,a))$ simulates the protocol $R(u,\pi_i(u))$, where $\pi_i(u)$ is guaranteed to be the cloud labeled $\pi(u,a)_1$.

The simulation of $R$ on $Z$ proceeds by emulating the transfer of some message across an edge $e=(v,w)$ in $G$, using a message transfer from the cloud of $v$ to the cloud of $w$ in $Z$. For simplicity imagine the case when no edges are corrupted. 
In the composed protocol, a vertex $(u,a)$ begins by sending its message $m$ to every vertex in its cloud using the protocols of $H$ (recall that each cloud is a copy of $H$). After that step, each one of the vertices $\{(u,b)\}_{b\in V(H)}$ holds the message $m$, and next the first round of $R$ is implemented. More precisely, suppose that in the protocol $R$ there is a message transmission from $u$ to $v$ along the edge $e$. Then in the simulation over $Z$, there is a transmission from the cloud of $u$ to the cloud of $v$: the message $m$ is first passed from $(u,e_1)$ to $(v,e_2)$ using the parallel edges between the vertices $(u,e_1)$ and $(v,e_2)$ (which are the vertices in the clouds that are associated to $e$), and then the message is propagated in the cloud of $v$ using the protocols of $H$ again. Thus, at the end of the first round, each cloud $C_v$ receives the message that $v$ was supposed to receive from $u$ in the first round. Subsequent rounds of $R$ are simulated in a similar manner. To argue correctness, one observes that if there are no corruptions, then each cloud $C_w$ acts ``unanimously''
(in the sense that all vertices in it are in full agreement) and hence can be thought of as a single entity, so that the correctness follows from the correctness of $R$.

Now suppose that an $\eps$-fraction of all the edges between clouds are corrupted, and no edges inside clouds are corrupt. 
In this case the clouds also act unanimously, and therefore we can analyze it by invoking the tolerance guarantees of $G$. Specifically, if the adversary corrupts more than $1/2$-fraction of the parallel edges between two neighboring clouds $C_v$ and $C_w$, it corresponds to an adversarial corruption of the edge $(v,w)$ in $G$. By Markov's inequality there can be at most a $2\eps$-fraction of such edges in $G$. Then, $R((u,a),\pi(u,a))$ is a simulation of $R$ under these edge corruptions, where each cloud behaves unanimously during the simulation. The tolerance guarantees of $G$ imply that most protocols $R(u,\pi_i(u))$ succeed, giving us that most of the protocols $R((u,a),\pi(u,a))$ succeed.

Let us finally discuss the general case, in which an $\eps$-fraction of all the edges of $Z$ are corrupted. This could include edges inside clouds too and we need to be more careful while analyzing the simulation. In particular, if at most $\eps$ fraction of edges of $Z$ are corrupted, most clouds will contain at most $\sqrt{\eps}$-fraction bad edges. The tolerance guarantees of $H$ now imply that at most $\poly(\eps)$-fraction of the vertices in each such cloud are ``doomed'', that is, the rest of the $1-\poly(\eps)$-fraction of vertices can perfectly communicate with each other. Given this, the clouds will behave almost-unanimously during the simulation of $R$ -- each incoming/outgoing message at $v$ in $R$, would be held by the non-doomed vertices in $C_v$. Thus if $R$ succeeds, then by the simulation guarantee, most vertices in the cloud of $\pi(u,a)_1$ would receive the correct message, and finally $\pi(u,a)$ can take a majority vote among these to get $m$. Therefore, we can still invoke the tolerance of $G$ to get that most of the protocols $R((u,a),\pi(u,a))$ succeed.

This completes the informal description of the protocol on $Z$ and its fault-tolerance; we defer the reader to Section~\ref{sec:composition} for a more formal presentation.

\subsubsection{Using Composition to Obtain Sparse Networks}
With the composition result in hand, we now explain how to prove Theorem~\ref{thm:main}. We start off with an $n$-vertex regular graph $G$ from~\cite{BMV24} that admits efficient tolerant protocols, where the degree of $G$ is $\Delta = {\sf poly}(\log n)$. We take a copy of $G$, which we call $\tilde{G}$, on $\Delta$ vertices\footnote{In reality we take $\tilde{G}$ on $\tilde{n}$ number of vertices, where $\Delta\leq \tilde{n}\leq O(\Delta)$, since that is the guarantee we have from the algebraic constructions of HDX.} thus, the degree of $\tilde{G}$ is $\tilde{\Delta} = {\sf poly}(\log\log n)$. Composing $G$ and $\tilde{G}$ gives a graph $G'$ on $n\Delta$ vertices that admits efficient, tolerant protocols and has degree $O(\tilde{\Delta}) = {\sf poly}(\log\log n)$. Repeating the composition one more time yields a graph $G''$ on $\Theta(n\Delta\tilde{\Delta})$ vertices with 
degree $d = {\sf poly}(\log\log\log n)\leq \log\log n$.

Finally, we take a constant degree expander graph $H$ on $D$ vertices with sufficiently good spectral expansion, as considered in~\cite{Upfal}. Upfal showed that such graphs $H$ admit protocols that are tolerant against constant fraction of edge corruptions and furthermore each vertex in $H$ has work complexity at most $2^{O(D)} = {\sf poly}(\log n)$. We 
invoke one final composition step, composing $G''$ with $H$, to get a constant degree graph that admits  efficient and tolerant protocols.

	\section{Preliminaries}
\paragraph{Notations:}
For functions $f,g\colon\mathbb{N}\to [0,\infty)$, we denote $f\lll g$ or $f = O(g)$ if there exists an absolute constant $C>0$ such that 
$f\leq C g$. Similarly, we denote $f\ggg g$ or $f = \Omega(g)$ if there exists an absolute constant $C>0$ such that 
$f\geq C g$.

\subsection{Routing Protocols under Adversarial Corruptions}
We start the discussion by formally defining communication protocols on a graph $G$.
\begin{definition}\label{def:comm-protocol}
Given a graph $G$, a routing protocol $P$ on $G$ is a set of rules where at each round, every vertex receives  messages from its neighbors, performs an arbitrary computation on all the messages received so far, and then based on the computation, sends some messages to its neighbours in $G$. There are two main parameters of interest:
\begin{enumerate}
\item Round complexity of $P$: This is the number of rounds of communication in $P$ and is denoted by $\text{round}(P)$.
\item Work complexity of $P$: The work done by a vertex in $P$ is the computation it performs throughout the protocol, as measured by the circuit size for the equivalent Boolean functions that the vertex computes. The work complexity of $P$, denoted $\text{work}(P)$, is the total work done by all the vertices to execute $P$.
\end{enumerate}
\end{definition}

Given this definition of communication over a graph $G$, we next define the two routing models that we study in this paper and the notion of fault-tolerance for them. 
We consider adversarial corruptions of edges -- we say an edge $(u,v) \in G$ is uncorrupted if whenever $u$ transfers a message $\sigma$ across $(u,v)$, then $v$ receives $\sigma$; otherwise, we say the edge $(u,v)$ is corrupted. 

The first model we consider is the almost-everywhere reliable transmission problem~\cite{Dwork}.
\begin{definition}\label{def:ae-transmission}
In the almost-everywhere reliable transmission model, the goal is to design a set of routing protocols $\cP=\{P(u,v)\}_{u,v\in V(G)}$ for message transmission between all pairs of vertices of $G$, where $P(u,v)$ transmits a message between $u$ and $v$. The parameters of interest are,
\begin{enumerate}
    \item Round complexity of $\cP := \max_{u,v}\text{round}(P(u,v))$.
    \item Work complexity of $\cP := \max_{u,v}\text{work}(P(u,v))$.     
    \item Error Tolerance: The set of protocols $\cP$ is said to be $(\eps,\nu)$-vertex (edge) tolerant if the following holds. Suppose an adversary corrupts any set of at most $\eps$-fraction of the vertices (edges) of $G$. Then there exists a set $S$ of at least $(1-\nu)$-fraction of vertices $G$, such that any two vertices $u,v \in S$ can reliably transmit messages between each other using $P(u,v)$. The set of vertices outside $S$ will be referred to as doomed nodes.\footnote{Here we have adapted the notion of doomed vertices for the edge-corruption model. In the context of vertex corruptions, that is, when all corrupt edges are concentrated on a few vertices, our notion of doomed vertices includes the corrupted vertices as well as those which are not corrupted by the adversary, but cannot communicate nevertheless.}
\end{enumerate}
\end{definition}

We also consider the more refined permutation routing model, which is relevant to PCPs~\cite{BMV24}.
\begin{definition}[Permutation Model]\label{def:perm-model}
In the permutation routing model, given a permutation $\pi: V(G) \rightarrow V(G)$, every vertex $u$ wishes to send a message to $\pi(u)$. The goal is to design a set of routing protocols $\cP=\{P(u,\pi(u))\}_{u\in V(G)}$, where $P(u,\pi(u))$ transmits a message between $u$ and $\pi(u)$. The parameters of interest are,
\begin{enumerate}
    \item Round complexity of $\cP := \max_u \text{round}(P(u,\pi(u)))$.
    \item Work complexity of $\cP:= \max_u\text{work}(P(u,\pi(u)))$.    
    \item Error Tolerance: The set of protocols $\cP$ is $(\eps,\nu)$-tolerant if the following holds. Suppose an adversary corrupts any set of at most $\eps$-fraction of the edges of $G$. Then at least $(1-\nu)$-fraction of the protocols $P(u,\pi(u))$ can reliably transmit a message between $u$ and $\pi(u)$.
\end{enumerate}
\end{definition}

	\section{Composition of Protocols using the Replacement Product}\label{sec:composition}
In this section, we let $G$ and $H$ be regular unweighted graphs (possibly with multi-edges) with $|V(H)|=\deg(G)=d$ that both admit efficient fault-tolerant routing protocols. We show that the graph $Z = G\rcirc H$, as defined in~\Cref{sec:rep-product}, also admits an efficient fault-tolerant protocol.

\subsection{The Protocol on $Z$}
Consider the graph $G,H$ and fix a set of protocols $\cP_H=\{P_{a,b}\}_{a,b\in V(H)}$ between all pairs of vertices of $H$. Fix a permutation $\pi$ on $V(Z)$; we will build a set of protocols $\cR=\{R((u,a),\pi(u,a))\}$ to route $\pi$. 

Our first step is to decompose $\pi$ into $d$ permutations $\pi_1,\ldots, \pi_d$ on $V(G)$. Specifically, we break the set of pairs $\{(u,a),\pi(u,a)\}_{(u,a)\in V(Z)}$ into $d$ sets, $S_1, \ldots, S_d$. Each set $S_i$ contains $n$ pairs such that their projection on the first coordinate forms a permutation $\pi_i$ on $G$, i.e. $\pi(u,a)_1 = \pi_i(u)$ for all $(u,a)\in S_i$. Now, for each permutation $\pi_i$ we consider the set of routing protocols $\cP_i = \{R(u,\pi_i(u))\}$ that route $\pi_i$ on $G$. Then for every vertex $(u,a)$ in $S_i$, the protocol $R((u,a),\pi(u,a))$,  ``simulates the protocol'' $R(u,\pi_i(u))$, denoted by $R$ henceforth.

In the introduction, we explained informally how one
``simulates a protocol $R$'' that is built for $G$, on the graph $Z$, and we now we make this more formal.
The simulation on $Z$ uses a sub-procedure called the cloud-to-cloud protocol, which at a high level, passes a message vector from a cloud $C_u$ to a cloud $C_v$, to emulate a message transfer from $u$ to $v$ across an edge $(u,v)$ in $G$. At the first step in our protocol, recall that $(u,a)$ sends its message $m$ to each vertex in $C_u$ using the protocols $\cP_H$ on $H$. Then the first round of $R$ is implemented-- each vertex in $C_u$, computes the message that $u$ would have sent across an edge $(u,v)$ (according to $R$), given the message it has received from $(u,a)$. This can be represented by  an outgoing message vector on $C_u$. In the ideal scenario, this vector equals $m'$ on every coordinate, where $m'$ is the message that $u$ sends to $v$ in the first round of $R$ when instantiated with $m$. This message vector is now transferred to $C_v$, using the cloud-to-cloud protocol, which we discuss in more detail later on. The subsequent rounds of $R$ are simulated in a similar way. That is, at any round $t$, for every cloud $C_v$, each vertex $(v,b)$ in it, first aggregates all the messages it received from its neighbors in the previous round, as part of the simulation of $R$. It computes the outgoing message that $R$ would have sent across any incident edge $(v,w)$ on this transcript and this is represented by an outgoing message vector on the cloud $C_v$. This vector then gets transferred to the cloud $C_{w}$ via the cloud-to-cloud protocol.

To formalize this protocol we set up some notations. Let $E_v(G)$ be the set of edges that are incident on $v$ in $G$, and let $2E(G)$ represent the set of pairs $(v, e)$ where $e \in E_v(G)$. We specify a protocol $R$ using a set of functions $\{\out_t: 2E(G)\times \Sigma^* \rightarrow \Sigma^* \cup \{\perp\}\}_{t\leq T}$, where $\out_t(v,e,\sigma)$ denotes the message that a vertex $v$ sends along an edge $e$ at time $t$,  given the transcript $\sigma$ of the previous rounds. The message received at $v$ at the end of the protocol is denoted by $\out_T(v,\sigma)$ and $\out_T$ is also part of the functions above. 
To simulate $R$ on $Z$, at every vertex $(v,b)$ we maintain the transcript of the protocol using the strings $\{\text{history}_t(v,b)\}$ and the outgoing messages it sends with respect to the edge $e \in E_v(G)$, denoted by $\out_t((v,b),e)$. At any time step $t$, we have an outgoing message vector on the cloud $C_v$, denoted by $[\out_t((v,b),e)]$, to be sent across the edge $e$, to the cloud $C_w$ using the cloud-to-cloud protocol from $C_v$ to $C_w$. The incoming message vector on $C_w$ in the next round is denoted by $[\In_{t+1}((w,c),e)]$. 

Given the cloud-to-cloud protocol, the routing protocol for $Z$ proceeds as follows.
\begin{mdframed}
\begin{algo}[Generate Protocols on $Z$ to Route a Permutation on $V(Z)$]\mbox{}\label{algo:zz-protocol}\\
\textbf{Input:} 
\begin{enumerate}
    \item A graph $G$ with an algorithm that given any permutation $\pi$ on $V(G)$, constructs the set of protocols $\cP(\pi)=\{P(u,\pi(u))\}_{u\in V(G)}$. 
    \item A graph $H$ with $|V(H)|=\deg(G)$ and a set $\cP_H$ of routing protocols $\{P(a,b)\}_{a,b\in V(H)}$. 
    \item A permutation $\pi: V(Z)\rightarrow V(Z)$.
\end{enumerate}
\textbf{Output:} A set of routing protocols $\cR=\{R((u,a),\pi(u,a))\}$ on $Z$.
\begin{enumerate}
    \item Create a bipartite graph $B= (L\sqcup R, E)$, with the left and right vertex sets, $L=R=V(G)$ and for every pair $((u,a),\pi(u,a))$, add an edge between $u$ and $\pi(u,a)_1$ (the vertex associated to $\pi(u,a)$ in $G$) with the label $((u,a),\pi(u,a))$.
    \item  Decompose $B$ into $d=\deg(G)$ matchings,
    each represented by a permutation $\pi_1,\ldots,\pi_d$ on $V(G)$. Let $S_i$ be the set of edge labels $((u,a),\pi(u,a))$ corresponding to the matching $\pi_i$.
    \item Let $\cP_i$ be the set of protocols that route $\pi_i$ on $G$, that is, $\cP_i = \{R(u,\pi_i(u))\}_{u\in G}$, with $T$ denoting the maximum round complexity of all the protocols.
    Associate to every pair $((u,a),\pi(u,a))$ the routing protocol $R(u,\pi_i(u))$ from $\cP_i$, where $i$ is the unique index for which $((u,a),\pi(u,a))$ belongs to the set $S_i$; note that $\pi(u,a)_1 = \pi_i(u)$. 
    \item For each vertex $(u,a)$, output the following protocol,
\end{enumerate}
\textbf{Protocol $R((u,a),\pi(u,a))$ for message transmission from $(u,a)$ to $\pi(u,a)$:}  
\begin{enumerate} 
\item At every vertex $(v,b)\in V(Z)$ initialize the strings: $\{\text{history}_t(v,b)\}_{t\leq T}$ and\\ $\{\out_t((v,b),e)\}_{e \in E_v(G),t\leq T}$, to empty strings; these will eventually store the transcript and the outgoing messages of the protocol at $(v,b)$. 
\item Set $\text{history}_1(u,a)=m$, the message that $(u,a)$ wishes to send. Then for each $b\in H$, send $m$ via the protocol $P(a,b)$ from $\cP_H$, and set $\text{history}_1(u,b)$ to the message received.
\item Let $R$ be the routing protocol on $G$ that is associated to $((u,a),\pi(u,a))$, specified by the functions $\{\out_t: 2E(G)\times \Sigma^* \rightarrow \Sigma^* \cup \{\perp\}\}_{t\leq T}$. Simulate the protocol $R$ as follows:\\
For every round $t \in \{1,\ldots, T\}$:
\vspace{-1ex}
\begin{enumerate}
\item If $t=1$, then skip to step (b). \\
Else, for every vertex $(v,b)$ in $Z$ and $e \in E_v(G)$, \\
Set $\text{history}_{t}(v,b) = \text{history}_{t-1}(v,b)\circ \In_{t}((v,b),e)$, \\
where $[\In_{t}((v,b),e)]_{b\in H}$ is the incoming message vector received at round $t$ on the cloud $C_v$ with respect to the edge $e$. 
\item If $t=T$ then for every vertex $(\pi_i(u),b)$ in $C_{\pi_i(u)}$:\\ 
set the message received as, $\out_{T}((\pi_i(u),b)) =\out_{T}(\pi_i(u),\text{history}_{T}(\pi_i(u),b))$, then end the protocol.\\
Else, for every vertex $(v,b)$ in $Z$ and edge $e\in E_v(G)$, \\
set $\out_{t}((v,b),e) =\out_{t}(v,e,\text{history}_{t}(v,b))$.
\item For every vertex $v\in G$ and edge $e\in E_v(G)$,
run the cloud-to-cloud protocol with input $e$, the graph $H$ with protocols $\cP_H$ and message vector $[\out_t((v,b),e)]_{b\in C_v}$. 
\end{enumerate}
\end{enumerate}
\end{algo}
\end{mdframed}
We now discuss the cloud-to-cloud protocol. At a high level, in this protocol we simulate a message transfer across an edge $e$ of $G$ with endpoints $v,w$ as a message transfer from $C_v$ to $C_w$. We do so by using the $\deg(H)$ parallel edges between $(v,e_1)$ and $(w,e_2)$,  the vertices in the clouds of $v$ and $w$ that are associated to the edge $e$. Thus the cloud to cloud transfer takes the form
\[C_v \rightarrow (v,e_1) \rightarrow (w,e_2) \rightarrow C_v,\]
where each arrow above denotes a propagation using majority votes. More precisely, the cloud-to-cloud protocol proceeds as follows.
\begin{mdframed}
\begin{protocol}[Cloud to Cloud Protocol]\mbox{}\label{algo:cloud-to-cloud}\\
\textbf{Input:} The graphs $G$ and $H$, with a set of protocols between all pairs of vertices $\cP_H=\{P(a,b)\}_{a,b\in H}$, an edge $e\in G$, that is associated to $(v,e_1)$ in $C_v$ and $(w,e_2)$ in $C_w$ and a message vector $M_{in}: C_v \to \Sigma$. \\
\textbf{Output:}  A message vector $M_{out}: C_w \to \Sigma$.
\begin{enumerate}
\item Every vertex $(v,b)\in C_v$, transmits the message $M_{in}(v,b)$ to $(v,e_1)$ via the protocols $P(b,e_1)$.
\item The vertex $(v,e_1)$ takes a majority of the messages it receives from vertices in $C_v$, and transmits these messages to $(w,e_2)$ via the parallel edges between $(v,e_1)$ and $(w,e_2)$. 
\item The vertex $(w,e_2)$ takes a majority of the messages it receives from $(v,e_1)$ and transmits this message to every vertex $(w,c)\in C_w$ via the protocol $P(e_2,c)$, that sets $M_{out}(w,c)$ as the message it receives.
\end{enumerate}
\end{protocol}
\end{mdframed}
This completes the formal description of the routing protocol over $Z$.
\subsection{Analysis of Protocols Generated by~\Cref{algo:zz-protocol}}
Each protocol $R((u,a),\pi(u,a))$ is a simulation of some protocol $R=R(u,\pi_i(u))$ on $G$, where every message transfer on an edge in $G$ is simulated by a cloud-to-cloud transfer in $Z$. As discussed in the introduction, it is easy to see that these protocols succeed in simulating the protocols on $G$ correctly if there are no edge corruptions. To address the general case where edge corruptions are present, we think of each edge $e$ in $G$ as a ``super-edge''. Under a set of edge corruptions on $Z$, we say that a super-edge $e = (v,w)$  is good if the corresponding cloud-to-cloud transfer from $C_v$ to $C_w$ always succeeds; otherwise, we say $e$ is bad. When we say 
the ``cloud-to-cloud transfer is successful'', we roughly mean that if a majority of vertices in $C_v$ set their outgoing message as $\sigma$, then a majority of vertices in $C_w$ will receive $\sigma$ at the end of~\Cref{algo:cloud-to-cloud}. We first show that if at most $\eps$-fraction of edges in $Z$ are bad, then at most $O(\sqrt{\eps})$-fraction super-edges are bad, denoted by the set $\cS$. We then show that if the protocol $R(u,\pi_i(u))$ succeeds in transferring a message from $u$ to $\pi_i(u)$  under any adversarial behavior of the edges in $\cS$, then $R((u,a),\pi(u,a))$ also succeeds in transferring a message from $(u,a)$ from $\pi(u,a)$. Once we have this, since the measure of $\cS$ is small, the tolerance guarantees of $G$ imply that only few of the protocols $R(u,\pi_i(u))$ are unsuccessful. This in turn implies that only a few of the protocols $R((u,a),\pi(u,a))$ are unsuccessful.
The core of the analysis of \Cref{algo:zz-protocol} is captured by the following lemma.

\begin{lemma}\label{lem:composition}
There is $c>0$ such that the following holds. Suppose that $G$ and $H$ are regular graphs as above satisfying that:
\begin{enumerate}
    \item The graph $G$ has an $(\eps_1,\nu_1)$-edge-tolerant routing protocol for every permutation on $V(G)$ with work complexity $W_1$ and round complexity $R_1$,
    \item The graph $H$ has $\deg(G)$ vertices, and  there is a collection of protocols $\cP_H = \{P(a,b)\}_{a,b\in V(H)}$ between all pairs of vertices, such that for any adversary corrupting at most $\eps_2$-fraction of $E(H)$:
    \begin{enumerate}
        \item At most $\nu_2$-fraction of the vertices of $V(H)$ are doomed.
        \item The set of protocols $\cP$ have work and round complexity $W_2$ and $R_2$ respectively.
    \end{enumerate}
\end{enumerate}
If $\nu_2 \leq c$, then for every permutation $\pi$ on $V(Z)$, \Cref{algo:zz-protocol} produces $\cR=\{R((u,a),\pi(u,a))\}$ a set of protocols on the graph $Z=G \rcirc H$ that are $(\eps,\nu)$-tolerant, for all $\eps \lll \min(c, \eps_2^2, (\eps_1-O(\nu_2))^2)$ and $\nu \lll O(\sqrt{\eps}+\nu_1+\nu_2)$. All protocols in $\cR$ have work complexity $O(W_1 W_2)$ and round complexity $O(R_1 R_2)$. Furthermore, if the protocols on $G$ and $H$ are polynomial time constructible then~\Cref{algo:zz-protocol} also runs in polynomial time.
\end{lemma}
The rest of this section is devoted for the proof of Lemma~\ref{lem:composition}.

Fix a permutation $\pi: V(Z)\rightarrow V(Z)$. As shown in the first two steps of \Cref{algo:zz-protocol}, $\pi$ corresponds to $d=\deg(G)$ many permutations $\pi_1,\ldots,\pi_d$ on $V(G)$. We will now prove that the protocols generated are tolerant against a constant fraction of edge-corruptions in $Z$.

\vspace{-2ex}
\paragraph{The outer protocol over clouds:} As in~\Cref{algo:zz-protocol}, fix the set of routing protocols $\cP_i = \{P(u,\pi_i(u)))\}$ over $G$, each of round complexity $R_1$ and work complexity $W_1$.

\vspace{-2ex}
\paragraph{The inner protocol inside clouds:} 
We fix the set of protocols  $\cP_{H} = \{P(a,b)\}_{a,b\in H}$, as in \Cref{algo:zz-protocol}, between all pairs of vertices $v,w$ in $H$. Since the graph inside every cloud $C_u$ is isomorphic to $H$, we can use these protocols for message transfers between the vertices of $C_u$.

\vspace{-2ex}
\paragraph{Bad super-edges and bounding them:}
We now begin the error analysis of \Cref{algo:zz-protocol}, and for that we need to introduce a few notions. First, fix any set of corrupted edges $\cE \subset E(Z)$ with measure at most $\eps$, which satisfies the assumptions in the lemma statement. A cloud $C_v$ is called bad if it contains too many corrupted edges, and more precisely, if $\Pr_{e\sim E(C_v)}[e \in \cE] \geq \sqrt{2\eps}$; otherwise, we say $C_v$ is good. Since the edges inside clouds constitute at least $1/2$ of the overall number of edges, by Markov's inequality at most $\sqrt{2\eps}$-fraction of the clouds $C_v$ are bad.

Let us now fix a set of doomed vertices for every cloud $C_v$, that is, the smallest set of vertices $\cD_v$ such that for every pair of vertices $a,b\notin \cD_v$, the internal protocol $P(a,b)$ can successfully transfer a message from $(v,a)$ to $(v,b)$. If $C_v$ is a good cloud we know that the  fraction of corrupted edges inside it is at most $\sqrt{2\eps}$ which is smaller than $\eps_2$. Since the induced subgraph over any cloud $C_v$ is isomorphic to $H$, the tolerance guarantees of $H$ imply that at most $\nu_2$-fraction of the protocols $P(a,b)$ fail, which gives that the fractional size of $\cD_v$ is at most $\nu_2$.

A super-edge $e\in E(G)$ with endpoints $v,w$ is said to be corrupted if either $C_v$ or $C_w$ is a bad cloud, $(v,e_1)\in \cD_v$,  $(w,e_2)\in \cD_w$, or at least $\sqrt{2\eps}$-fraction of the (parallel) edges between $(v,e_1)$ and $(w,e_2)$ are corrupted. Using Markov's inequality and a union bound we get that   
\begin{equation}\label{eq:corrupted_super_edge}
\Pr_{e \sim E(G)}[e \text{ is corrupted}] \lll \eps^{1/2}+\nu_2.  
\end{equation}

\vspace{-2ex}
\paragraph{Cloud to Cloud Transfer on a Good Super-edge:} the following claim asserts that the cloud-to-cloud transfer works well on an uncorrupted super edge. More precisely:
\begin{claim}\label{claim:good-super-edge}
Fix any uncorrupted super-edge $e\in E(G)$ with endpoints $v,w$. Suppose the cloud-to-cloud transfer in \Cref{algo:cloud-to-cloud} is run with the edge $e$ and a message vector $M_{in}: C_v \to \Sigma$ satisfying  $M_{in}(b)=\sigma$ for all $b \in C_v\setminus \cD_v$. Then the output of the protocol is $M_{out}: C_w \to \Sigma$ satisfying $M_{out}(c)=\sigma$ for all $c \in C_w \setminus \cD_w$.    
\end{claim}

\begin{proof}
Let $e$ correspond to the vertices $(v,e_1)$ in $C_v$ and $(w,e_2)$ in $C_w$. The proof of this claim is broken into three message transfers: from $C_{v}$ to $(v,e_1)$, from $(v,e_1)$ to $(w,e_2)$, and finally from $(w,e_2)$ to $C_w$, and we argue about each step separately. Firstly, if $e_1$ is not in $\cD_v$, then it will receive the value $\sigma$ from all vertices $C_v \setminus \cD_v$, which is at least $1-\nu_2 \geq 1/2$. Therefore $(v,e_1)$ will compute the correct majority and set its outgoing message as $\sigma$. 

Next, for the message transfer between $(v,e_1)$ and $(w,e_2)$, since at most $\sqrt{\eps}$-fraction of the edges between them are corrupted, the vertex $(w,e_2)$ will receive $\sigma$ on at least $1/2$-fraction of the edges, thus setting its outgoing message as $\sigma$.

Finally for the message transfer between $(w,e_2)$ to $C_w$, every vertex in $C_w \setminus \cD_w$ receives the message $\sigma$ from $(w,e_2)$ since $e_2 \notin \cD_w$, thus proving the claim. 
\end{proof}

\vspace{-2ex}
\paragraph{Analysis of the Protocol $R((u,a),\pi(u,a))$:} Let $\cS$ be the set of corrupted super-edges. Suppose that $((u,a),\pi(u,a))$ is associated to the protocol $R(u,\pi_i(u))$. We argue that the protocol $R((u,a),\pi(u,a))$ is a correct simulation of the protocol $R(u,\pi_i(u))$ on $G$ when run with the set of corrupted edges $\cS$. Formally, we show:
\begin{claim}\label{claim:simulation}
Fix a set $\cE\subseteq E(Z)$ and let $\cS\subseteq E(G)$ be the set of corrupted super edges as above. Suppose the vertices $(u,a)$ and $\pi(u,a)$ satisfy that $a \notin \cD_u$,  $\pi(u,a)_2 \notin \cD_{\pi_i(u)}$ and the protocol $R(u,\pi_i(u))$ is successful in the message transfer from $u \rightarrow \pi_i(u)$ 
with corruption set $\cS$ under any adversarial strategy on it. Then \Cref{algo:zz-protocol} is successful in transferring a message from $(u,a)$ to $\pi(u,a)$ with corruption set $\cE$ under any adversarial behavior on it.
\end{claim}

\begin{proof}
Henceforth, suppose that $a\notin \cD_u$ and let $m$ be the message that $(u,a)$ wishes to send to $\pi_i(u,a)$. Our analysis will compare the transcript of $R((u,a),\pi(u,a))$ with the transcript of $R(u,\pi_i(u))$, denoted henceforth by $R$, and show that they are equal up to the adversarial behavior of the edges in $\cS$. To do so, we will define a protocol $R_\cS$ which is the same as $R$, except the messages sent across $\cS$ are ``erased''. Then we will show that the transcript of $R((u,a),\pi(u,a))$ matches the transcript of $R_\cS$ on the locations that haven't been erased.

Let us start by defining $R_\cS$. In this protocol, the message that $v$ sends across any edge $e$ in $\cS$ is $\star$, which is then received as $\star$ on the other end. Furthermore, the outgoing message that $v$ sends on an edge $e$ not in $\cS$ is computed by taking into account the erasures on its transcript from the previous rounds -- it sends $\sigma$ on $e$, if on every setting of the erased indices in its transcript, $R$ would have sent out $\sigma$; otherwise $v$ sends a $\star$. Formally, let the protocol $R=R(u,\pi_i(u))$ be specified by the functions $\{\out_t: 2E(G) \times \Sigma^* \to \Sigma^* \cup \{\perp\}\}$, that is, $\out_t(v,e,H)$ is the message that the vertex $v\in G$ sends across the edge $e \in E_v(G)$ at round $t$ when given the transcript $H$ of the previous rounds. We extend the definition of the function $\out_t(v,e,H)$ to the setting where some values in $H$ might be $\star$'s: $\out_t(v,e,H)$ is defined to be $\sigma$ if it equals $\sigma$ for all fixings of the $\star$-locations, and $\star$ otherwise. Let $\text{history}_t(v)$ denote the transcript of $R_\cS$ (when run on the initial message $m$ at $u$) at $v$ up to round $t$, and let $M_t(v,e)$ denote the message that $v$ sends across $e$ at round $t$ in $R_\cS$. Concretely we have the following inductive definition for these strings, 
\[\text{history}_1(v)=\begin{cases}
m ~~~~\text{if $v=u$,} \\
\emptyset ~~~\text{otherwise.}
\end{cases},
\]   
\[
\text{history}_t(v)=\text{history}_{t-1}(v)\circ_{e \in E_v(G) \text{ with endpoint }w} M_{t-1}(w,e),
\]
and finally for $t\neq T$,
\[M_t(v,e)=\begin{cases}
\star ~~~~~~~~\text{if } e \in \cS,\\  
\out_t(v,e,\text{history}_t(v)) ~~~ \text{otherwise,}
\end{cases}
\] 
and 
\[M_T(\pi_i(u))=\out_t(\pi_i(u),\text{history}_T(\pi_i(u))).
\]

Next, we argue that the transcript of $R((u,a),\pi(u,a))$ matches the transcript of $R_\cS$ on the non-$\star$ locations. Specifically, consider the transcript $\text{history}_t(v,b)$ and the outgoing messages $\out_t(v,e,\text{history}_t(v,b))$ at any vertex $(v,b)$, where $b\notin \cD_v$, from the protocol $R((u,a),\pi(u,a))$ in~\Cref{algo:zz-protocol}. We argue inductively that for all rounds $t$, $\text{history}_t(v,b)$ equals $\text{history}_t(v)$ on all the locations that are not $\star$. Furthermore, the outgoing message $\out_t(v,e,\text{history}_t(v,b))$ equals $M_t(v,e)$, when the latter is not equal to $\star$. 

{\bf Base case:} For $t=1$, we only need to check the claim for vertices in $C_u$, since every other vertex starts out with an empty transcript. Since $a \notin \cD_u$, it is the case that every $b\in C_u$ that is also not doomed receives $m$ from $(u,a)$ and sets $\text{history}_1(u,b)=m$ which also equals $\text{history}_1(u)$, as required. Such a vertex sets its outgoing message $\out_1((u,b),e)$ on an edge $e$, to $\out_1(u,e,m)$. When $e\notin \cS$, by definition this equals $M_1(u,e)$, thus proving the claim for $t=1$. 

{\bf Inductive step:} We now assume that the hypothesis holds for all rounds $t \leq j$ and prove it for $t=j+1$. Fix a vertex  $(v,b)$ where $b\notin \cD_v$. By the inductive hypothesis, we know that $\text{history}_{j}(v,b)=\text{history}_j(v)$, where the equality holds for all coordinates where the latter string is not equal to $\star$. For any edge $e$ with endpoints $v,w$ that is not in $\cS$ and  $M_j(w,e)\neq \star$, by the inductive hypothesis, the cloud-to-cloud protocol for $C_{w} \rightarrow C_v$ is instantiated with the vector $M=[\out_j((w,c),e)]$ satisfying $M[c]=M_j(w,e)$ for every $c \notin \cD_{w}$. Since $e \notin \cS$, \Cref{claim:good-super-edge} implies that the cloud-to-cloud transfer is successful, and in particular, the vertex $(v,b)$ receives the message $M_j(w,e)$, which gives us that $\text{history}_{j+1}(v,b)=\text{history}_{j+1}(v)$ on the locations that are not $\star$. Now we know that the outgoing message that $(v,b)$ sends on any edge $e$ is $\out_{j+1}((v,b),e,\text{history}_{j+1}(v,b))$ and is set to $\out_t(v,e,\text{history}_{j+1}(v,b))$. Since the transcript at $(v,b)$ is equal to the transcript at $v$ on all non-$\star$ locations, when $\out_t(v,e,\text{history}_{j+1}(v))$ is not equal to $\star$, then $\out_t(v,e,\text{history}_{j+1}(v,b))=\out_t(v,e,\text{history}_{j+1}(v))$. This in turn equals $M_{j+1}(v,e)$, when the latter is not equal to $\star$, thus proving the inductive claim.

Now note that by the assumption in the lemma, $R$ is successful in the message transfer from $u$ to $\pi_i(u)$, regardless of the adversarial behavior of edges in $\cS$. Therefore the received message $M_T(\pi_i(u))$ in the protocol $\cR_S$ is equal to $m$.
Using the inductive claim above we get that the transcripts of $R((u,a),\pi(u,a))$ match $\cR_S$ on the non-$\star$ locations, and in particular at the final round on $\pi_i(v)$. That is, every vertex $(\pi_i(v),b)$ where $b \notin \cD_{\pi_i(v)}$, sets $\out_T(\pi_i(v),b)$ as $m$. Since $\pi(u,a)_2$ is not in $\cD_{\pi_i(v)}$ we get that $\out_T(\pi_i(v),\pi(u,a)_2)$ equals $m$, implying that the vertex $\pi(u,a)$ successfully receives $m$ at the end of the protocol.
\end{proof}

\paragraph{Bounding the fraction of failed transmissions from $(u,a)$ to $\pi(u,a)$:}
First note that the distribution over the protocol on $G$ that is associated to the pair $((u,a),\pi(u,a))$ for a uniformly random $(u,a)\sim Z$, is the distribution over protocols $R(u,\pi_i(u))$ for uniformly random $i\sim [d],u\sim G$. Secondly, by~\eqref{eq:corrupted_super_edge} we have $\Pr_{e\in E(G)}[e\in \cS]\lll \sqrt{\eps}+\nu_2\leq \eps_1$. By the tolerance guarantees of $G$, we get that for every $i$, at most $\nu_1$-fraction of the protocols $\{R(u,\pi_i(u))\}_{u\in G}$ fail when run with the corrupted set of edges $\cS$. We can now apply \Cref{claim:simulation} to bound the fraction of failed transmissions between $(u,a)$ and $\pi(u,a)$ using a union bound,
\begin{align*}
&\Pr_{(u,a)\in Z}[\text{the message transfer from }(u,a) \rightarrow \pi(u,a) \text{ fails}]\\
&\leq \Pr_{(u,a)\in Z}[a \in \cD_u]+ \Pr_{(u,a)\in Z}[\pi(u,a)_2 \in \cD_{\pi(u,a)_1}] + \Pr_{i\in [d], u \in V(G)}[R(u,\pi_i(u))\text{ fails with corruption set }\cS]\\
&\lll \eps^{1/2}+\nu_1+\nu_2.
\end{align*}
This completes the proof of Lemma~\ref{lem:composition}.
\qed

\begin{remark}\label{rem:composition}
\Cref{lem:composition} as stated composes graphs $G$ and $H$ where the number of vertices in $H$ is equal to the degree of $G$. This lemma can easily be adapted to the setting where the number of vertices in $H$ is at least $\deg(G)$ and at most $O(\deg(G))$, 
 and our application requires this extension. The argument is essentially the same but involves some notational inconvenience. In this case, we modify the graph product so that each cloud contains $|V(H)|$ vertices, but only $\deg(G)$ of them are associated with edges of $G$. In particular, the ratio between the number of cross cloud edges and inner cloud edges is no longer $1$, and is instead $O(1)$. This leads to the fraction of bad super-edges being larger by a constant factor, leading to a change in the tolerance parameters by constant factors too. 
\end{remark}

	\section{Routing Network with Constant Tolerance and Constant Degree}
In this section we use Lemma~\ref{lem:composition} to prove Theorem~\ref{thm:main}.

\subsection{A Routing Network with Polylogarithmic Degree}
We first state the construction of the routing network from~\cite[Lemma D.3]{BMV24} that is based on the high-dimensional expanders constructed in~\cite{LSV1, LSV2}. 
\begin{theorem}\label{thm:routing-polylog}
For all $n \in \mathbb{N}$ there exists a regular graph  $G = (V, E)$ (with multiedges) on $\Theta(n)$ vertices with degree $\pl n$ such that for all permutations $\pi$ on $V(G)$, there is a set of routing protocols $\cR=\{R(u,\pi(u))\}_{u\in G}$ that is $(\eps,O(\eps))$-edge-tolerant for all $\eps > 0$, with round complexity $O(\log n)$ and work complexity $\pl n$. Furthermore, both the graph and the routing protocols can be constructed in time $\poly(n)$. 
\end{theorem}

For the composition we need a set of protocols between all pairs of vertices of $G$, such that if some small fraction of edges behave adversarially, then only few vertices are doomed. This is an immediate corollary of~\Cref{thm:routing-polylog}.
\begin{proposition}\label{prop:perm-to-all-pairs}
Suppose that $G$ is a graph such that for every permutation $\pi$, there is an $(\eps,\nu)$-edge-tolerant protocol with $\nu \leq 0.01$, work complexity $W$ and round complexity $R$. Then one can construct in polynomial time a set of protocols $\cP_G=\{P(a,b)\}_{a,b\in G}$ with work complexity $W|V(G)|$ and round complexity $O(R)$, such that if at most $\eps$-fraction of the edges of $G$ are corrupted, then at most $\sqrt{\nu}$-fraction of the vertices of $G$ are doomed. 
\end{proposition}

\begin{proof}
We can assume without loss of generality that $m=|V(G)|$ is even. Thus we can decompose the edge set of the complete graph on $m$ vertices into $m$ matchings/permutations $\pi_1,\ldots,\pi_m$ on $V(G)$. For every matching $\pi_i$, we let $\cR_i = \{R(u,\pi_i(u))\}_{u\in G}$ be the set of protocols routing it. 

Let $P(u,\pi_i(u))$ be the protocol that first sends $u$'s message to $w$ via $R(u,w)$ for every $w\in V(G)$ and then every $w$ sends the message received to $\pi_i(u)$ via the protocol $R(w,\pi_i(u))$. Finally $\pi_i(u)$ takes a majority over all the messages it receives. Let $\cP_G = \cup_{i\in [m]} \{P(u,\pi_i(u))\}_{u\in G}$.

It is easy to check that the protocol $P(u,\pi_i(u))$ has work complexity $W|V(G)|$, so let us analyze the fraction of doomed vertices. If at most $\eps$-fraction of the edges are corrupted then at most $\nu$-fraction of the protocols $\{R(u,v)\}_{u,v\in V(G)}$ fail. Let $\cD$ be the set of vertices $v$ for which more than $\sqrt{\nu}$-fraction of the protocols $\{R(v,w)\}$ fail. One can check that if $u,v \notin \cD$ then the protocol $P(u,v)$ successfully transfers a message from $u$ to $v$, since $v$ at least $1-2\sqrt{\nu} \geq 1/2$-fraction of the messages that reach $v$ are correct, leading to it computing the correct majority.
\end{proof}
Note that the above proposition gives us a set of fault-tolerant all pair protocols on the $N$ vertex graph from \Cref{thm:routing-polylog} with a work complexity of $N\pl N$, instead of $\pl N$. This loss is affordable though since in the composition we will only apply this set of all pair protocols on a graph with $N=\pl n$ vertices, that is, on the smaller graph in the composition.

\subsection{A Routing Network with Constant Degree and Exponential Work Complexity}
We also need~\cite[Theorem 2]{Upfal}, who showed that  constant degree expander graphs have fault-tolerant routing protocols with work complexity at most $\exp(n)$. Formally,
\begin{theorem}\label{thm:routing-constant}  
There is a universal constant $d\in \N$ such that for all $n \in \mathbb{N}$, there exists a $d$ regular graph $G = (V, E)$ on $n$ vertices such that in $\poly(n)$-time one can construct  a set of protocols $\cP_G=\{P(a,b)\}_{a,b\in G}$ with work complexity $\exp(n)$ and round complexity $O(\log n)$, such that if at most $\eps$-fraction of the edges of $G$ are corrupted, then at most $O(\eps)$-fraction of the vertices of $G$ are doomed. 
\end{theorem}
Note that~\cite{Upfal}'s theorem is stated for vertex corruptions, but the statement above for edge corruptions follows by noting that any degree $d$ network that is tolerant to $\eps$-fraction of vertex corruptions is also tolerant to $\eps/d$-fraction of edge corruptions, which suffices since $d$ is a fixed constant.

\subsection{Proof of \Cref{thm:main}}\label{sec:perm-final}
We now prove~\Cref{thm:main}, restated below.

\begin{lemma}[\Cref{thm:main} restated]\label{thm:main-restated}
There exists $D\in \N$ such that for all $\eps>0$, for large enough $n$, there exists a $D$-regular graph $G$ with $\Theta(n)$ vertices such that for each permutation $\pi$ on $V(G)$, the graph $G$ admits a set of $(\eps^{32},O(\eps))$-edge-tolerant routing protocols $\mathcal{R}=\{R(u,\pi(u))\}$ with work complexity $\pl n$ and round complexity $\tO(\log n)$. Furthermore, both the graph and protocols can be constructed in time $\poly(n)$. 
\end{lemma}

\begin{proof}
We assume $\eps\leq c$ where $c$ is a small absolute constant, otherwise the statement is vacuously true. Let $G_1$ be the graph from \Cref{thm:routing-polylog} on $\Theta(N)$ vertices ($N\in \N$ to be chosen later) and degree $D_1=\pl N$ that admits an $(\eps,O(\eps))$-edge-tolerant routing protocol with work complexity $\pl N$. Let $G_2$ be another graph from \Cref{thm:routing-polylog} on $N_2$ vertices, for $D_1 \leq N_2 \leq O(D_1)$, and degree $\pl D_1$ that admits a $(\Theta(\eps^2),\Theta(\eps^2))$-edge-tolerant routing protocol with work complexity $\pl D_1$. Applying \Cref{prop:perm-to-all-pairs} for $G_2$, we get a set of protocols between all pairs of vertices of $G_2$, such that if at most $\Theta(\eps^2)$-fraction of the edges of $G_2$ are corrupted, then at most $O(\eps)$-fraction of its vertices are doomed.
Now composing $G_1$ and $G_2$ using \Cref{lem:composition} (in the strengthened setting where $|V(H)|\leq O(D_1)$, see \Cref{rem:composition})  we get a graph $G_3 = G_1 \rcirc G_2$ on $O(ND_1) = N\pl N$ vertices, with degree $\pl D_1 = \pll N$. We also get that for every permutation on $V(G_3)$, there is a set of routing protocols that are $(\eps_3,\nu_3)$-tolerant for $\eps_3 = \Theta(\min(\eps^4, \eps^2))=\Theta(\eps^4)$ and $\nu_3=\Theta(\sqrt{\eps_2}+\eps)=\Theta(\eps)$. The protocols have work complexity $\pl N$ and round complexity $O(\log N\log\log N)$.

Similarly, composing $G_3$ with another graph $G_4$ from \Cref{thm:routing-polylog} on $N_2$ vertices, with $\deg(G_3)\leq N_2 \leq O(\deg(G_3))$, we get a graph $G_5$ on $N\pl N$ vertices, with degree $\text{polylogloglog} N$. We also get that for every permutation on $V(G_5)$, there is a set of routing protocols that are $(\eps_5,\nu_5)$-tolerant for $\eps_5 = \Theta(\eps^{16})$ and $\nu_5=\Theta(\eps)$. The protocols have work complexity $\pl N$ and round complexity $\tO(\log N)$.

Finally, we can compose $G_5$ with a constant degree graph $G_6$ from \Cref{thm:routing-constant} on $N_6$ vertices and degree $D$ (a fixed constant), with $N_6=\deg(G_5)$, to get a graph $G_7$ that has $N\pl N$ vertices, degree $D$ and admits $(\Theta(\eps^{32}),\Theta(\eps))$-edge-tolerant protocols for every permutation on $V(G_7)$. These protocols have work complexity $\pl N \cdot 2^{\text{polylogloglog} N}\leq \pl N$ and round complexity $\tO(\log N)$.

Taking $G=G_7$, $N=n/\pl n$, gives us a constant degree graph on $\Theta(n)$ vertices that has $(\eps^{32},\eps)$-edge-tolerant protocols with work complexity $\pl n$ and round complexity $\tO(\log n)$, for every permutation.
\end{proof}

\subsection{Proof of \Cref{cor:main}}\label{sec:cor-final}
One could apply~\Cref{prop:perm-to-all-pairs} to~\Cref{thm:main-restated} to get a result for the almost-everywhere reliable transmission problem with a constant fraction of doomed nodes, but the work complexity would be $n\pl n$. To get a result with polylogarithmic work complexity we instead use a simple randomized construction of protocols, utilizing the protocols from~\Cref{thm:main-restated} along with a majority.

\begin{corollary}[\Cref{cor:main} restated]\label{cor:main-restated}
There exists $D\in \N$ such that for all $\eps>0$, for large enough $n$, there exists a $D$-regular graph $G$ with $\Theta(n)$ vertices and a set of protocols $\cR=\{R(u,v)\}_{u,v\in G}$ between all pairs of vertices in $G$, with work and round complexity $\pl n$, such that if at most $\eps^{32}$-fraction of edges are corrupted, then at most $O(\eps)$-fraction of vertices in $G$ are doomed.
Furthermore, there is a deterministic algorithm that computes $G$ in time $\poly(n)$ and a randomized algorithm that constructs the protocols in time $\poly(n)$ that satisfy the above tolerance guarantees with probability $1-\exp(-n\pl n)$.  
\end{corollary}

\begin{proof}
Let $G$ be the graph from~\Cref{thm:main-restated} on $\Theta(n)$ vertices. Assume $n$ is even without loss of generality, and break the edges of the complete graph on $n$ vertices into a union of $n$ matchings. For each matching, apply~\Cref{thm:main-restated} to get a set of routing protocols that are $(\eps^{32},O(\eps))$-tolerant to edge corruptions. This gives us a set of routing protocols $\cR'=\{R'(u,v)\}_{u,v\in G}$ between all pairs of vertices in $G$, such that if at most $\eps^{32}$-fraction of edges are corrupted then at most $O(\eps)$-fraction of protocols $R'(u,v)$ fail. We now show how to upgrade this guarantee to the stronger tolerance guarantee where there are only an $O(\eps)$-fraction of doomed vertices.

For each pair of vertices $(u,v)$, pick a random set of vertices $S_{u,v}$ of size $\pl n$ from $G$. Let $R(u,v)$ be the protocol where $u$ sends its message to each vertex $w$ in $S_{u,v}$ via the protocol $R'(u,w)$, and then every vertex $w$ sends the message it received to $v$ via $R'(w,v)$, and finally $v$ takes a majority vote over all the messages received. Note that the protocol $R(u,v)$ and $R(v,u)$ may be different; this choice helps make the analysis cleaner.

To analyze the tolerance guarantees of $\cR = \{R(u,v)\}$, first fix a set of corrupted edges $\cE$ with measure $\eps^{32}$ in $G$. By the above, we know that at most $O(\eps)$-fraction of the protocols $R'(u,v)$ fail, each of which is called a bad protocol. Let $D_1$ the set of vertices $u$ for which at least $1/8$-fraction of $R'(u,v)$ are bad. Using Markov's inequality we know that the size of $|D_1|\leq O(\eps n)$. 

Fix a vertex $u$ not in $D_1$, and consider the random set $S_{u,v}$ for any $v$. Call this set corrupted if more than $1/4$-fraction of the protocols, $R'(u,w)$ for $w\in S_{u,v}$ are bad. Using a Chernoff bound followed by a union bound over $v$ in $G$, we get that,
\[\Pr_{\cR}[\exists v, S_{u,v}\text{ is corrupted}]\leq \exp(-\pl{n}).\]
Now create a set $D_2$ of vertices $u$ for which the above bad event happens, that is, there exists $v$, such that  $S_{u,v}$ is corrupted.  The above bound says that every vertex not in $D_1$, belongs to $D_2$ with probability at most $\exp(-\pl n)$ and furthermore each of these events is independent. Therefore using the multiplicative Chernoff bound we get that,
\[\Pr_{\cR}[|D_2| \ggg \eps n]\leq \exp(-n\pl{n}).\]
Similarly let $D_3$ be the set of vertices $v\notin D_1$, for which there exists $u$ such that $S_{u,v}$ is corrupted. Using the same analysis as above we get that the size of $D_3$ is also at most $O(\eps n)$ with probability $1-\exp(-n\pl n)$. Taking $D(\cE)$ to be the union of $D_1$, $D_2$ and $D_3$, it is easy to check that every two vertices $w,w'$ outside $D$ can communicate perfectly via $R(w,w')$ and $R(w',w)$ and,
\[\Pr_{\cR}[|D(\cE)| \ggg \eps n]\leq \exp(-n\pl{n}).\]
Finally, we can apply a union bound over all possible sets of adversarial edges $\mathcal{E}$, which are at most $2^{O(n)}$ in number since $G$ is a constant-degree graph. This gives that, with probability at least $1 - \exp(-n\pl n)$, the set of protocols $\mathcal{R}$ constructed as above has at most an $O(\eps)$-fraction of doomed vertices regardless of which $\eps^{32}$-fraction of edges the adversary chooses to corrupt.
\end{proof}

	\section*{Acknowledgments}
	We thank Nikhil Vyas for many helpful discussions, and regret that he declined co-authoring this paper.
	
	\bibliographystyle{alpha}
	\bibliography{references}

\end{document}